\documentclass[a4paper,11pt]{article}
\pdfoutput=1
\usepackage{fullpage}
\usepackage{hyperref}
\hypersetup{
      colorlinks=false,
      linkcolor=blue,
      urlcolor=cyan,
}
\date{}
\usepackage{authblk}
\usepackage[noadjust]{cite}
\usepackage{amsfonts}

\usepackage{algpseudocode}
\usepackage{amsmath}
\usepackage{comment}
\usepackage{subcaption}
\usepackage[textwidth=20mm]{todonotes}
\usepackage{tikz}
\usepackage{algorithm}

\usetikzlibrary{calc,decorations,decorations.markings,decorations.pathmorphing,intersections,math,positioning,snakes}

\pgfdeclaredecoration{complete sines}{initial} 
{
    \state{initial}[
        width=+0pt,
        next state=sine,
        persistent precomputation={\pgfmathsetmacro\matchinglength{
            \pgfdecoratedinputsegmentlength / int(\pgfdecoratedinputsegmentlength/\pgfdecorationsegmentlength)}
            \setlength{\pgfdecorationsegmentlength}{\matchinglength pt}
        }] {}
    \state{sine}[width=\pgfdecorationsegmentlength]{
        \pgfpathsine{\pgfpoint{0.25\pgfdecorationsegmentlength}{0.5\pgfdecorationsegmentamplitude}}
        \pgfpathcosine{\pgfpoint{0.25\pgfdecorationsegmentlength}{-0.5\pgfdecorationsegmentamplitude}}
        \pgfpathsine{\pgfpoint{0.25\pgfdecorationsegmentlength}{-0.5\pgfdecorationsegmentamplitude}}
        \pgfpathcosine{\pgfpoint{0.25\pgfdecorationsegmentlength}{0.5\pgfdecorationsegmentamplitude}}
}
    \state{final}{}
}

\tikzset{
	vertex/.style={minimum size=#1,circle,fill=black,draw,inner sep=0pt},
	decoration={markings,mark=at position .5 with {\arrow[black,thick]{stealth}}},
	vertex/.default=2.5mm,
	bigVertex/.style={vertex=1mm},
	e0/.style={line width=0.8pt},
	e1/.style={line width=2.4pt},
	em/.style={line width=0.8pt,decoration={
            complete sines,
            segment length=8,
            amplitude=5
        },
        decorate},
    em0/.style={line width=0.8pt,decoration={
            complete sines,
            segment length=12,
            amplitude=3
        },
        decorate},
	transformsTo/.pic=
	{
		\coordinate (-leftEnd) at (0,0);
		\coordinate (-rightEnd) at (5,0);
		\draw[thick] (0,0) -- (0,1) -- (3,1) -- (3,2) -- (5,0) -- (3,-2) -- (3,-1) -- (0,-1) -- (0,0);
	},
	pics/red/.style n args={5}{code=
	{
		\tikzmath{\a=0.5;\b=0.6;}
		\node (v) at (0,\a)			[bigVertex,label=right:$v$]{};
		\node (w) at (0,-\a)		[bigVertex,label=right:$w$]{};
		\node (a) at (-\a,2*\b)		[bigVertex,label=above:$a$]{};
		\node (b) at (\a,2*\b)		[bigVertex,label=above:$b$]{};
		\node (c) at (-\a,-2*\b)	[bigVertex,label=below:$c$]{};
		\node (d) at (\a,-2*\b)		[bigVertex,label=below:$d$]{};
		
		\draw[#1] (v) -- (w);
		\draw[#2] (v) -- (a);
		\draw[#3] (v) -- (b);
		\draw[#4] (w) -- (c);
		\draw[#5] (w) -- (d);
	}},
	pics/red-straight/.style n args={2}{code=
	{
		\tikzmath{\a=0.5;\b=0.6;}
		\node (a) at (-\a,2*\b)		[bigVertex,label=above:$a$]{};
		\node (b) at (\a,2*\b)		[bigVertex,label=above:$b$]{};
		\node (c) at (-\a,-2*\b)	[bigVertex,label=below:$c$]{};
		\node (d) at (\a,-2*\b)		[bigVertex,label=below:$d$]{};
		
		\draw[#1] (a) -- (c);
		\draw[#2] (b) -- (d);
	}},
	pics/red-cross/.style n args={2}{code=
	{
		\tikzmath{\a=0.5;\b=0.6;}
		\node (a) at (-\a,2*\b)		[bigVertex,label=above:$a$]{};
		\node (b) at (\a,2*\b)		[bigVertex,label=above:$b$]{};
		\node (c) at (-\a,-2*\b)	[bigVertex,label=below:$c$]{};
		\node (d) at (\a,-2*\b)		[bigVertex,label=below:$d$]{};
		
		\draw[#1] (a) -- (d);
		\draw[#2] (b) -- (c);
	}},
	pics/red-ii/.style n args={5}{code=
	{
		\tikzmath{\a=0.5;}
		\node (a) at (0,3*\a)		[bigVertex,label=right:$a$]{};
		\node (v) at (0,\a)			[bigVertex,label=right:$v$]{};
		\node (w) at (0,-\a)		[bigVertex,label=right:$w$]{};
		\node (b) at (0,-3*\a)		[bigVertex,label=right:$b$]{};
		
		\draw[#1] (a) -- node [midway,left] {#5} (v);
		\draw[#2] (v) edge [bend right] (w);
		\draw[#3] (v) edge [bend left] (w);
		\draw[#4] (w) -- (b);
	}},
	pics/red-ii-single/.style n args={2}{code=
	{
		\tikzmath{\a=0.5;}
		\node (a) at (0,3*\a)		[bigVertex,label=right:$a$]{};
		\node (b) at (0,-3*\a)		[bigVertex,label=right:$b$]{};
		
		\draw[#1] (a) -- node [midway] {#2} (b);
	}},
	pics/swap/.style n args={4}{code=
	{
		\useasboundingbox (-1,-2.5) rectangle (1,2.6);
		\tikzmath{\a=1.5;}
		\begin{scope}[shift={(0,-0.5*\a)}]
		\foreach \i in {0,...,5}
			\node (v\i) at (\i*60:\a)		[bigVertex]{};
		\node (u0) at (-0.5*\a,2*\a)		[bigVertex]{};
		\node (u1) at (0.5*\a,2*\a)			[bigVertex]{};
		\node (z0) at (1.5*\a,1.4*\a) 		[bigVertex]{};
		\node (z1) at (1.5*\a,0.6*\a) 		[bigVertex]{};
		
		\draw[e1,#4] (v0) edge [bend right] (v1);
		\draw[#2,#4] (v1) -- node [midway,above,black] {$e'$} (v2);
		\draw[e1,#4] (v2) edge [bend right] (v3);
		\draw[e1,#3] (v3) edge [bend right] (v4);
		\draw[#1,#3] (v4) -- node [midway,below,black] {$e$} (v5);
		
		\draw[e1,#3] (v5) edge [bend right] (v0);
		\draw[e1,red] (v0) edge [bend right] (z1);
		\draw[e0] (z1) -- node [midway,right] {$f'$} (z0);
		\draw[e1,red] (z0) edge [bend right] (u1);
		\draw[e1,red] (u1) -- node [midway,above,black] {$f$} (u0);
		\draw[e1,red] (u0) edge [bend right] (v3);
		\end{scope}
	}},
	case3a/.pic=
	{
		\useasboundingbox (-1,-1.3) rectangle (1,1);
		\pic (r-) {red={e0}{e1}{e1}{e1}{e0}};
		\draw[e1] (r-a) edge [bend right] (r-c);
		\draw[e1] (r-a) edge [out=180,in=235,looseness=2.5] (r-d);
	},
	case3b-1/.pic=
	{
		\pic (r-) {red={e0}{e1}{e1}{e1}{e0}};
		\draw[e1] (r-a) edge [bend right] (r-c);
		\draw[e1] (r-b) edge [bend left] (r-d);
	},
	case3b-2/.pic=
	{
		\pic (r-) {red-cross={e1}{e0}};
		\draw[e1] (r-a) edge [bend right] (r-c);
		\draw[e1] (r-b) edge [bend left] (r-d);
	},
	case2a-1/.pic=
	{
		\useasboundingbox (-1,-2) rectangle (1,2);
		\pic (r-) {red={e0}{e1}{e0}{e1}{e0}};
		\draw[e1] (r-a) edge [bend left] node [vertex=0mm,pos=0.5,name=z1] {} (r-b);
		\draw[e1] (r-c) edge [bend right] node [vertex=0mm,pos=0.5,name=z2] {} (r-d);
		\draw[e1] (z1) edge[out=45,in=-45,looseness=3] (z2);
	},
	case2a-2/.pic=
	{
		\useasboundingbox (-1,-2) rectangle (1,2);
		\pic (r-) {red-straight={e0}{e0}};
		\draw[e1] (r-a) edge [bend left] node [vertex=0mm,pos=0.5,name=z1] {} (r-b);
		\draw[e1] (r-c) edge [bend right] node [vertex=0mm,pos=0.5,name=z2] {} (r-d);
		\draw[e1] (z1) edge[out=45,in=-45,looseness=3.5] (z2);
	},
	case2b-1/.pic=
	{
		\useasboundingbox (-1,-2) rectangle (1,2);
		\pic (r-) {red={e0}{e1}{e0}{e1}{e0}};
		\draw[e1] (r-a) edge [bend right] node [vertex=0mm,pos=0.5,name=z1] {} (r-c);
		\draw[e1] (r-b) edge [bend left] node [vertex=0mm,pos=0.5,name=z2] {} (r-d);
		\draw[e1] (z1) edge[out=135,in=45,looseness=5] (z2);
	},
	case2b-2/.pic=
	{
		\useasboundingbox (-1,-2) rectangle (1,2);
		\pic (r-) {red-cross={e0}{e0}};
		\draw[e1] (r-a) edge [bend right] node [vertex=0mm,pos=0.5,name=z1] {} (r-c);
		\draw[e1] (r-b) edge [bend left] node [vertex=0mm,pos=0.5,name=z2] {} (r-d);
		\draw[e1] (z1) edge[out=135,in=45,looseness=5] (z2);
	},
	case2c-1/.pic=
	{
		\useasboundingbox (-1,-2) rectangle (1,2);
		\pic (r-) {red={e0}{e1}{e0}{e1}{e0}};
		\draw[e1] (r-a) edge [out=-135,in=180] node [vertex=0mm,pos=0.4,name=z1] {} (r-d);
		\draw[e1] (r-b) edge [out=-45,in=0] node [vertex=0mm,pos=0.4,name=z2] {} (r-c);
		\draw[e1] (z1) -- (z2);
	},
	case2c-2/.pic=
	{
		\useasboundingbox (-1,-2) rectangle (1,2);
		\pic (r-) {red-straight={e0}{e0}};
		\draw[e1] (r-a) edge [out=-135,in=180] node [vertex=0mm,pos=0.4,name=z1] {} (r-d);
		\draw[e1] (r-b) edge [out=-45,in=0] node [vertex=0mm,pos=0.4,name=z2] {} (r-c);
		\draw[e1] (z1) -- (z2);
	},
	pics/invariant/.style n args={3}{code=
	{
		\useasboundingbox (-1,-1.75) rectangle (1,2.6);
		
		\tikzmath{\a=0.5;\b=0.6;}
		\pic (r-) {#1};
		
		\node (z1) at (-1.75*\a,0)		[bigVertex]{};
		\node (z2) at (1.75*\a,0)		[bigVertex]{};
		
		\draw[e1,red] (r-a) edge [out=-135,in=90] (z1);
		\draw[e1] (r-c) edge [out=135,in=90] (z1);
		
		\draw[e1,#2] (r-b) edge [out=-45,in=90] (z2);
		\draw[e1,#3] (r-d) edge [out=45,in=-90] (z2);
		
		\node (u1) at (2*\a,4*\b)	[bigVertex]{};
		\node (u2) at (-2*\a,4*\b) 	[bigVertex]{};
		\node (u3) at (-3*\a,2.5*\b) 	[bigVertex]{};
		
		\draw[e1,red] (z2) edge [bend right] (u1);
		\draw[e0] (u1) -- node [midway,above] {$f$} (u2);
		\draw[e1,red] (u2) -- node [midway,above left,black] {$e$} (u3);
		\draw[e1,red] (u3) edge [bend right] (z1);
	}}
}

\DeclareMathOperator{\cover}{cover}
\DeclareMathOperator{\cost}{cost}
\DeclareMathOperator{\link}{link}
\DeclareMathOperator{\cut}{cut}
\DeclareMathOperator{\update}{update}
\DeclareMathOperator{\evert}{evert}
\DeclareMathOperator{\vertex}{\bf vertex}
\DeclareMathOperator{\real}{\bf real}
\DeclareMathOperator{\remove}{remove}
\DeclareMathOperator{\add}{add}
\DeclareMathOperator{\roottree}{root}
\DeclareMathOperator{\connected}{connected}
\DeclareMathOperator{\mincost}{mincost}

\newcommand{\Oh}{\mathcal{O}}

\title{Finding perfect matchings in bridgeless cubic multigraphs without dynamic (2-)connectivity}

\usepackage{amsthm}
\usepackage[T1]{fontenc}
\newtheorem{theorem}{Theorem}
\newtheorem{lemma}{Lemma}
\newtheorem{invariant0}{Invariant}

\author[ ]{Paweł Gawrychowski}
\author[ ]{Mateusz Wasylkiewicz\thanks{Partially supported by Polish National Science Center grant 2018/29/B/ST6/02633.}}
\affil[ ]{Institute of Computer Science, University of Wrocław, Poland\\ \texttt{\{gawry,mateusz.wasylkiewicz\}@cs.uni.wroc.pl}}

\begin{document}
\maketitle

\begin{abstract}
Petersen's theorem, one of the earliest results in graph theory, states that any bridgeless cubic multigraph contains a perfect matching.
While the original proof was neither constructive nor algorithmic, Biedl, Bose, Demaine, and Lubiw [J. Algorithms 38(1)] showed how
to implement a later constructive proof by Frink in $\Oh(n\log^{4}n)$ time using a fully dynamic 2-edge-connectivity structure. Then, Diks and
Stańczyk [SOFSEM 2010] described a faster approach that only needs a fully dynamic connectivity
structure and works in $\Oh(n\log^{2}n)$ time. Both algorithms, while reasonable simple, utilize non-trivial (2-edge-)connectivity structures. We show that this is not necessary,
and in fact a structure for maintaining a dynamic tree, e.g. link-cut trees, suffices to obtain a simple $\Oh(n\log n)$ time algorithm.
\end{abstract}

\section{Introduction}

Finding a maximum cardinality matching in a given graph is one of the fundamental algorithmic problems in graph theory. For bipartite graphs,
it can be seen as a special case of the more general problem of finding a maximum flow, which immediately implies a polynomial-time algorithm.
Already in the early 70s, Hopcroft and Karp~\cite{HopcroftKarp1973} obtained a fast $\Oh(m\sqrt{n})$ time algorithm for this problem, where $m$ denotes the number of
edges and $n$ the number of vertices. For general graphs, Edmonds~\cite{Edmonds1965} designed an algorithm working in $\Oh(mn^{2})$ time,
and in 1980 Micali and Vazirani~\cite{MicaliVazirani1980} stated an $\Oh(m\sqrt{n})$ time algorithm. For dense graphs, a better complexity
of $\Oh(n^{\omega})$, where $\omega$ is the exponent of $n\times n$ matrix multiplication, has been achieved by Mucha and Sankowski~\cite{MuchaSankowski2004}.
For the case of sparse graphs, i.e. $m=\Oh(n)$, 
a long and successful line of research based on applying continuous techniques resulted in an $m^{1+o(1)}$ time algorithm by Chen et al.~\cite{ChenEtAl2022}
for the bipartite case.
However, there was no further improvement for the general case, and the $\Oh(n^{1.5})$ time algorithm obtained by applying the approach
of Micali and Vazirani~\cite{MicaliVazirani1980} remains unchallenged.

This naturally sparked interest in searching for natural classes of sparse graphs that admit a faster algorithm. A natural candidate is a class
of graphs that always contain a perfect matching. One of the earliest results in graph theory attributed to Petersen~\cite{Petersen1891}, states that any bridgeless
cubic graph contains a perfect matching, where cubic means that the degree of every vertex is exactly 3, while bridgeless means that it is not possible to
remove a single edge to disconnect the graph. In fact, the theorem is still true for a cubic multigraph with at most two bridges, and from now on
we will consider multigraphs, i.e. allow loops and parallel edges. The original proof was very complicated and non-constructive, but 
Frink~\cite{Frink1926} provided another approach that can be easily implemented to obtain a perfect matching in $\Oh(n^2)$ time.
The high-level idea of this approach is to repeatedly apply one of the two possible reductions, in each step choosing the one that maintains
the invariant that the current multigraph is bridgeless and cubic. Then, we revert the reductions one-by-one, which possibly requires finding
an alternating cycle to make sure that a particular edge does not belong to the matching. 
Biedl, Bose, Demaine, and Lubiw~\cite{BiedlEtAl2001} improved the time complexity to $\Oh(n\log^4{n})$ thanks to two insights.
First, we can apply a fully dynamic 2-edge-connectivity structure of Holm, de Lichtenberg, and Thorup~\cite{HolmEtAl2001} to decide which reduction should be applied. Second,
finding an alternating cycle can be avoided by requiring that a chosen edge does not belong to the matching.
Diks and Stańczyk~\cite{DiksStanczyk2010} further improved the complexity to $\Oh(n\log^2{n})$ time by observing that in fact a fully dynamic connectivity
structure suffices if we additionally maintain a spanning tree of the current multigraph in a link-cut tree.
By plugging in the fully dynamic connectivity structure by Wulff-Nilsen~\cite{WulffNilsen2012}, the complexity of their algorithm
can be further decreased to $\Oh(n\log^{2}n/\log\log n)$. Alternatively, at the expense of allowing randomization and bit-tricks,
plugging in the structure of Huang, Huang, Kopelowitz, Pettie, and Thorup~\cite{HuangEtAl2023} results in expected $\Oh(n\log n (\log\log n)^2)$ running time.

\subparagraph*{Our contribution. } The algorithms of Biedl, Bose, Demaine, and Lubiw~\cite{BiedlEtAl2001} and Diks and Stańczyk~\cite{DiksStanczyk2010} can be seen as efficient implementations
of Frink's proof~\cite{Frink1926}. While both are reasonably simple (see Section~\ref{sec:outline}), they arguably hide some of their complexity in the
fully dynamic (2-edge-)connectivity structure. We show that in fact this can be avoided, and present an implementation that only needs a structure for
maintaining a dynamic tree, such as a link-cut tree. This results in a simple and self-contained algorithm that works in deterministic $\Oh(n\log n)$ time
without any bit-tricks.

\subparagraph*{Applications. }
Petersen's theorem can be generalized to cubic multigraphs with at most two bridges. Finding perfect matching in such multigraphs can be easily reduced
to finding perfect matching in bridgeless cubic multigraphs in linear time~\cite{BiedlEtAl2001}. Hence, our algorithm can be used for finding perfect matching
in cubic multigraphs with at most two bridges in $\Oh(n\log{n})$ time. We note that generalizing our algorithm to arbitrary cubic graphs can be difficult, as 
finding perfect matching in general graphs can be reduced in linear time to  finding perfect matching in cubic graphs~\cite{Biedl2001}.

The complement of a perfect matching of a cubic graph is a $2$-factor. A $2$-factor can be used to approximate graphic TSP problem, where we have to find
a shortest tour visiting all vertices of an undirected graph. Several approximation algorithms for the graphic TSP problem in cubic graphs were
presented~\cite{GamarnikEtAl2005,AggarwalEtAl2011,BoydEtAl2011,CorreaEtAl2015,Zuylen2015,CandrakovaLukotka2015,DvorakEtAl2016}. Recently,
Wigal, Yoo and Yu~\cite{WigalEtAl2021} presented a $5/4$-approximation algorithm for the graphic TSP problem in cubic graphs which works in $\Oh(n^2)$ time.
We can obtain a faster algorithm for this problem in bridgeless cubic graphs by patching the cycles of a $2$-factor of the input graph into a single tour. However,
its approximation ratio is $5/3$, since every cycle of the computed $2$-factor has length at least three. We can improve the approximation ratio to $3/2$ by
computing a $2$-factor with no cycles of length three. Finding such a 2-factor can be reduced in $\Oh(n)$ time to finding ordinary $2$-factor in bridgeless
cubic multigraphs by first repeatedly contracting the cycles of length three of the input graph. 
Thus, we obtain an $3/2$-approximation algorithm working in $\Oh(n\log{n})$ time.

Another application of finding perfect matching in bridgeless cubic graphs is finding $P_4$-decomposition in bridgeless cubic graphs, which consists in
partitioning the set of edges of the input graph into a collection of paths of length exactly three. Kotzig~\cite{Kotzig1957} presented a simple construction
that, given a perfect matching $M$ of a cubic graph, finds its $P_4$-decomposition in linear time by directing every cycle of the complement of $M$.
Therefore, using our algorithm, such a decomposition can be found in $\Oh(n\log{n})$ time in bridgeless cubic graphs.

\section{Preliminaries}
Let $G=(V,E)$ be an undirected and connected multigraph, i.e. it may contain loops and parallel edges. We denote the number of vertices of $G$ by $n$ and the number of edges by $m$.
We say that $G$ is \emph{cubic} if the degree of every vertex of $G$ is equal to three. 
We denote an edge connecting vertices $v$ and $w$ by $\{v,w\}$. If there is exactly one copy of an edge connecting $v$ and $w$, we refer to such an edge as \emph{single}.
Moreover, we call a pair of two different edges connecting the same pair
of vertices a \emph{double edge}. Given an edge $e$ of $G$, we denote by $G\setminus e$ the multigraph obtained from $G$ by removing exactly one copy of $e$. Given a
vertex $v$ of $G$, we denote by $G\setminus v$ the multigraph obtained from $G$ by removing vertex $v$ together with all its incident edges. Given an edge $e=\{v,w\}$, for
some $v,w\in V$, we denote by $G\cup e$ the multigraph obtained from $G$ by adding one copy of $e$.
An edge $e$ of $G$ is called a \emph{bridge} if its removal disconnects $G$. We say that $G$ is \emph{bridgeless} if no edge of $G$ is a bridge.

A multigraph $H=(V',E')$ is said to be a \emph{subgraph} of $G$ if $V'\subseteq V$ and $E'\subseteq E$. A subgraph $H$ of $G$ is said to be \emph{spanning} if $V=V'$. A spanning subgraph $T$ of $G$ is said to be a \emph{spanning tree} of $G$ if $T$ is a tree. We
denote the set of all vertices of a subgraph $H$ of $G$ by $V(H)$ and the multiset of all edges of $H$ by $E(H)$. A \emph{path} is a finite sequence of vertices and edges
$P=(v_0,e_1,v_1,e_2,v_2,\ldots,v_{\ell-1},e_\ell,v_\ell)$ of $G$, for some nonnegative integer $\ell$, where the vertices $v_0,v_1,\ldots,v_\ell\in V$ are pairwise distinct,
$e_1,e_2\ldots,e_\ell\in E$, and $e_i=\{v_{i-1},v_i\}$ for every $i\in\{1,\ldots,\ell\}$. We refer to $\ell$ as the \emph{length} of $P$. We say that $P$ \emph{connects} vertices $v_0$ and $v_\ell$. A \emph{cycle} is defined
similarly, except that $v_0$ and $v_\ell$ should be equal. We often identify a path or a cycle in $G$ with the subgraph of $G$ consisting of all its vertices and edges. Given
a spanning tree $T$ of $G$, we say that $e\in E(G)\setminus E(T)$ \emph{covers} $f\in E(T)$ (in $T$) if  $f\in E(P)$, where $P$ is the path in $T$ connecting
the endpoints of $e$.

A subset $M\subseteq E$ is said to be a \emph{matching} of $G$ if the degree of every vertex in the subgraph of $H=(V,M)$ 
is at most one. A matching $M$ of $G$ is said to be \emph{perfect} if the degree of every vertex in the subgraph $H=(V,M)$ is equal to one.
The \emph{perfect matching problem} consists in finding a perfect matching of a given multigraph, if it exists.
Given an edge $e$, we say that it is \emph{matched} (with respect to $M$) if it belongs to $M$. Otherwise, we say that it is \emph{unmatched}
(with respect to $M$).
An \emph{($M$-)alternating cycle}
is a cycle of $G$ whose edges alternately belong and do not belong to $M$.
An \emph{application} of an $M$-alternating cycle $A$ to $M$ is an operation that removes all matched edges of $A$ from $M$ and adds all
unmatched edges of $A$ to $M$.

\section{Link-cut trees}\label{sec:link-cut}

We need a structure for maintaining a forest of vertex-disjoint rooted trees, each of whose edges has a real-valued cost. Link-cut trees of Sleator and Tarjan~\cite{SleatorTarjan1981} support the following
operations (among others) in $\Oh(\log n)$ time each, where $n$ is the total number of vertices:
\begin{description}
\item $\roottree(\vertex v)$: return the root of the tree containing $v$.
\item $\cost(\vertex v)$: returns the cost of the edge from $v$ to its parent. We assume that $v$ is not a root.
\item $\mincost(\vertex v)$: returns the vertex $w$ closest to $\roottree(v)$ such that the edge from $w$ to its parent has minimum cost on the path connecting $v$ and $\roottree(v)$. We assume that $v$ is not a root.
\item $\update(\vertex v, \real x)$: add $x$ to the cost of every edge on the path connecting $v$ and $\roottree(v)$.
\item $\link(\vertex u, v, \real x)$: combine the trees containing $u$ and $v$ by adding an edge $(u,v)$ with cost $x$, making $v$ the parent of $u$. We assume that
$u$ and $v$ are in different trees, and $u$ is a root.
\item $\cut(\vertex v)$: delete the edge from $v$ to its parent. We assume that $v$ is not a root.
\item $\evert(\vertex v)$: modify the tree by making $v$ the root.
\end{description}
As mentioned in the original paper, instead of real-valued costs we can in fact work with an arbitrary (but fixed) semigroup. In particular, we can use the semigroup
$G=(E,\oplus)$, where $x\oplus y = x$ for any $x, y$. This allows us to maintain a forest of vertex-disjoint unrooted trees, each of whose edges $e$
has its associated label $\cover(e)$, under the following operations:
\begin{description}
\item $\connected(u,v)$: check if $u$ and $v$ belong to the same tree.
\item $\remove(u,v)$: remove an edge $\{u,v\}$ from the forest. We assume that the edge belongs to some tree.
\item $\add(u,v,x)$: add an edge $\{u,v\}$ to the forest, and set its label to be $x$. We assume that $u$ and $v$ are in different trees.
\item $\cover(u,v)$: return the label of the edge $\{u,v\}$. We assume that the edge belongs to some tree.
\item $\update(u,v,x)$: set the label of every edge on the path connecting $u$ and $v$ to be $x$. We assume that $u$ and $v$ belong to the same tree.
\end{description}
It is straightforward to implement these operations in $\Oh(\log n)$ time each by maintaining a link-cut tree, except that we use the semigroup $G$ instead of real-valued costs.
\begin{itemize}
\item $\connected(u,v)$ checks if $\roottree(u)=\roottree(v)$.
\item $\remove(u,v)$ first calls $\evert(v)$, and then $\cut(u)$.
\item $\add(u,v,x)$ proceeds by calling $\evert(u)$, and then $\link(u,v,x)$.
\item $\cover(u,v)$ first calls $\evert(v)$, and then returns $\cost(u)$.
\item $\update(u,v,x)$ is implemented by calling $\evert(v)$, and then $\update(u,x)$.
\end{itemize}

By maintaining another link-cut tree with real-valued costs we can also support checking if the paths
connecting $u$ with $v$ and $u'$ with $v'$ share a common edge in $\Oh(\log n)$ time (assuming that $u, v, u', v'$ all belong to the same tree).
The cost of each edge is initially $0$. To implement a query, we first call $\evert(u)$ and $\update(v,-1)$. This has the effect
of setting the cost of every edge on the path connecting $u$ with $v$ to $-1$. Then, we call $\evert(u')$ and
check if $\mincost(v')$ returns $-1$, which happens if and only if the path connecting $u'$ and $v'$ shares a common edge with the path
connecting $u$ and $v$. Finally, we call $\evert(u)$ again, and then $\update(v,1)$ to restore the costs.

We note that any other structure for maintaining dynamic trees, e.g. top trees, could be used here in place of link-cut trees.

\section{Outline of previous algorithms}\label{sec:outline}
In this section we present previous algorithms for the perfect matching problem in bridgeless cubic multigraphs.

\subsection{$\Oh(n^2)$ time algorithm based on Frink's proof}\label{subsec:frink}
Frink's proof of Petersen's theorem can be easily turned into an algorithm. It uses the following theorem:
\begin{theorem}[Frink]\label{thm:frink}
Let $G$ be any bridgeless cubic multigraph and $\{v,w\}$ any single edge of $G$. Let $\{a,v\}$ and $\{b,v\}$ be other edges of $G$ incident to $v$. Let $\{c,w\}$ and $\{d,w\}$ be other edges of $G$ incident to $w$. Define multigraphs $H_1=((G\setminus v)\setminus w)\cup\{a,c\}\cup\{b,d\}$ and $H_2=((G\setminus v)\setminus w)\cup\{a,d\}\cup\{b,c\}$ (see Figure~\ref{fig:red}).
Then both $H_1$ and $H_2$ are cubic and at least one of them is bridgeless.
\end{theorem}

\noindent We call the operation of producing $H_1$ (resp. $H_2$) from $G$ a \emph{straight} (resp. \emph{crossing}) \emph{reduction (of type I)} on $\{v,w\}$. We refer to
both straight and crossing reductions as \emph{reductions (of type I)}. We do not provide the proof of the above theorem, but stress that it will follow from the analysis
of our algorithm, making the result self-contained.

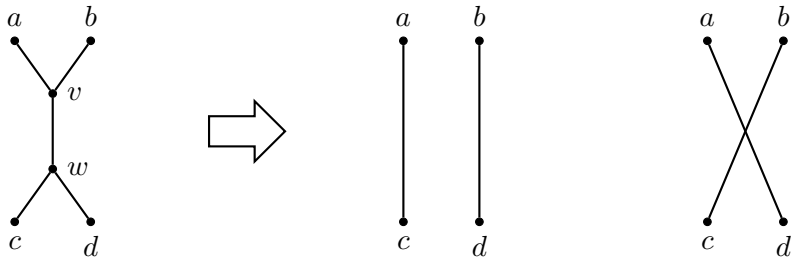
\begin{figure}[htpb]
\centering
\begin{tikzpicture}[transform shape]
	\pic (arrow) [scale=0.2]{transformsTo};
	\pic [left=2cm of arrow-leftEnd] {red={e0}{e0}{e0}{e0}{e0}};
	\pic [right=2cm of arrow-rightEnd] {red-straight={e0}{e0}};
	\pic [right=6cm of arrow-rightEnd] {red-cross={e0}{e0}};
\end{tikzpicture}
\caption{Straight and crossing reduction of type I on single edge $\{v,w\}$.\label{fig:red}}
\end{figure}

The idea of the algorithm based on Theorem~\ref{thm:frink} is to repeatedly perform the reduction on any single edge of the input multigraph $G_0$ to produce
a sequence of multigraphs $G_0$, $G_1$, \ldots, $G_k$ which are all cubic and bridgeless. It is easy to observe that every bridgeless cubic multigraph with more
than two vertices has a single edge. Hence, we can assume that $k=n/2-1$ and $|V(G_k)|=2$. To build a perfect matching of the input multigraph $G_0$, we can find any perfect
matching of $G_k$ and revert the reductions in a reverse order to find perfect matchings of $G_{k-1}$, $G_{k-2}$, \ldots, $G_0$. This can be done repeatedly
using the following theorem.

\begin{lemma}[\cite{BiedlEtAl2001}]\label{thm:reduction}
Let $G$ be any bridgeless cubic multigraph and $G'$ be a multigraph obtained by performing a reduction on a single edge of $G$. Given a perfect matching $M'$ of $G'$,
we can find a perfect matching of $G$ in $\Oh(n)$ time.
\end{lemma}
\begin{proof}
Without loss of generality, consider a straight reduction. We use the notation from the statement of Theorem~\ref{thm:frink}. We construct a perfect matching $M$ of $G$.
We start off with the empty set. We add every edge of $M'$ which belongs to $G$ to $M$. Hence, it remains to add $\{v,w\}$ or some edges incident to $\{v,w\}$ to $M$. We consider
the following three cases (see Figure~\ref{fig:reverting-reductions}).
\begin{enumerate}
\item[a)] If both $\{a,c\}$ and $\{b,d\}$ do not belong to $M'$, we add $\{v,w\}$ to $M$.
\item[b)] If either $\{a,c\}$ or $\{b,d\}$ belongs to $M'$, say $\{a,c\}$, we add $\{a,v\}$ and $\{c,w\}$ to $M$.
\item[c)] If both $\{a,c\}$ and $\{b,d\}$ belong to $M'$, we find and apply an $M'$-alternating cycle of $G'$ which contains $\{b,d\}$ to $M'$ using
Lemma~\ref{thm:alternating-cycle} below. Then we get either case~a) or~b).\qedhere
\end{enumerate}
\end{proof}

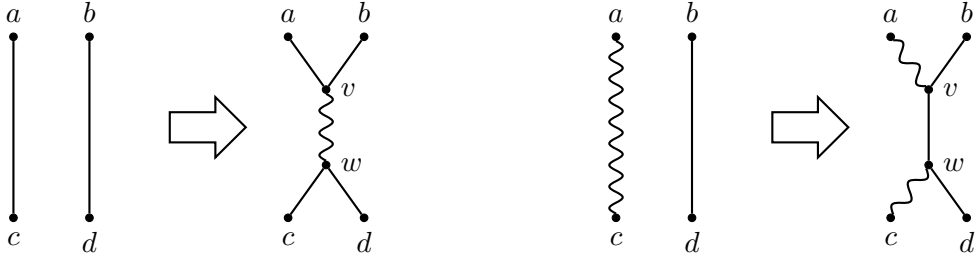
\begin{figure}[htpb]
\centering
\begin{subfigure}{.49\textwidth}
	\centering
	\begin{tikzpicture}[transform shape]
		\pic (arrow) [scale=0.2]{transformsTo};
		\pic [left=1.5cm of arrow-leftEnd] {red-straight={e0}{e0}};
		\pic [right=1cm of arrow-rightEnd] {red={em}{e0}{e0}{e0}{e0}};
	\end{tikzpicture}
\end{subfigure}
\begin{subfigure}{.49\textwidth}
	\centering
	\begin{tikzpicture}[transform shape]
		\pic (arrow) [scale=0.2]{transformsTo};
		\pic [left=1.5cm of arrow-leftEnd] {red-straight={em}{e0}};
		\pic [right=1cm of arrow-rightEnd] {red={e0}{em}{e0}{em}{e0}};
	\end{tikzpicture}
\end{subfigure}
\caption{Reverting reduction of type I. The matched edges are marked by wavy lines.\label{fig:reverting-reductions}}
\end{figure}

\begin{lemma}[\cite{BiedlEtAl2001}]\label{thm:alternating-cycle}
Let $G$ be a bridgeless cubic multigraph, $M$ a perfect matching of $G$, and $e$ an edge of $G$. Then $G$ has an $M$-alternating cycle that contains $e$
that can be found in $\Oh(n)$ time.
\end{lemma}

Using these lemmas we can find a perfect matching in bridgeless cubic multigraphs in $\Oh(n^2)$ time, since we can check if a multigraph is bridgeless and
find an alternating cycle in linear time.

\subsection{$\Oh(n\log^4{n})$ time algorithm with fully dynamic 2-edge-connectivity}\label{subsec:biedl}

Notice that there are two bottlenecks in the algorithm presented in Subsection~\ref{subsec:frink}: checking if a multigraph is bridgeless,
and finding an alternating cycle. To check if a multigraph is bridgeless, we can use a fully dynamic $2$-edge-connectivity structure.
Such a structure maintains a multigraph $G$ under the following operations:
\begin{itemize}
\item add an edge to $G$,
\item remove an edge from $G$,
\item check if given two vertices of $G$ are in the same $2$-edge-connected component of $G$.
\end{itemize}

To remove the first bottleneck, Biedl, Bose, Demaine, and Lubiw~\cite{BiedlEtAl2001} used the fully dynamic $2$-edge-connectivity structure given by Holm, de
Lichtenberg and Thorup~\cite{HolmEtAl2001} with $\Oh(\log^4{n})$ amortized time per operation, thus obtaining an $\Oh(n\log^4{n})$ time algorithm.
We note that plugging in a faster (and later) structure of Holm, Rotenberg and Thorup~\cite{HolmEtAl2018} improves the time complexity to $\Oh(n(\log{n})^2(\log{\log{n}})^2)$.

In order to remove the second bottleneck, the idea of Biedl, Bose, Demaine, and Lubiw~\cite{BiedlEtAl2001} was to forbid the case where both edges of some $E(G_i)\setminus E(G_{i-1})$ belong to the found perfect matching.
To this end, we choose any edge $e_0$ of the input multigraph $G_0$, and search for a perfect matching which does not contain $e_0$. Notice that
Lemma~\ref{thm:alternating-cycle} implies that such perfect matching always exists. Then, we perform a reduction on a single edge incident to $e_0$,
and we define $e_1$ as an edge of $E(G_1)\setminus E(G_0)$ such that $e_0\cap e_1\neq\emptyset$, that is, $e_{0}$ and $e_{1}$ are incident to the same vertex.
Next, we recursively find a perfect matching in $G_1$
which does not contain $e_1$. Again, we perform a reduction on a single edge incident to $e_1$, and so on. Recall that $G_k$ consists of exactly two vertices,
so it is trivial to find a perfect matching which does not contain $e_k$. Again, we revert all reductions to construct a perfect matching of the input multigraph $G_0$.
However, now we can use the assumption that the perfect matching of $G_i$ does not contain $e_i\in E(G_i)\setminus E(G_{i-1})$. Therefore, we can construct a
perfect matching of $G_{i-1}$ in constant time since we do not have to apply an alternating cycle to get rid of the case c) from the proof of Lemma~\ref{thm:reduction}.
This optimization gives us the desired $\Oh(n\log^4{n})$ running time.

Note that we cannot perform the reduction of type I if all edges of $G_i$ incident to $e_i$ are double edges. In such a case we perform a \emph{reduction of type II}
on any edge incident to $e_i$ instead of reduction of type I as follows. Consider a double edge $e=\{\{v,w\},\{v,w\}\}$ of $G_i$. Let $\{a,v\}$ be a single edge incident to
$v$ and $\{b,w\}$ a single edge incident to $w$. The reduction removes both copies of $\{v,w\}$ and all their incident edges, adds an edge $\{a,b\}$ to the multigraph
and defines $e_{i+1}=\{a,b\}$ (see Figure~\ref{fig:red-ii}). When reverting this reduction, we have a guarantee that $\{a,b\}$ does not belong to a perfect matching, hence we can add any copy of $\{v,w\}$ to it.

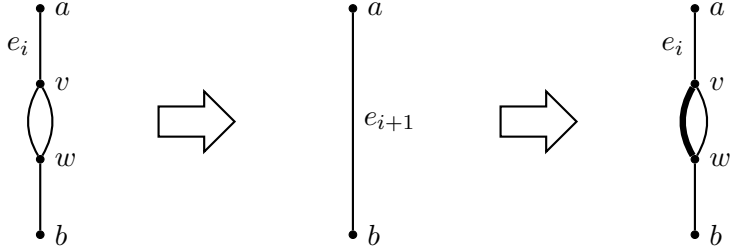
\begin{figure}[htpb]
\centering
\begin{tikzpicture}[transform shape]
	\pic (arrow1) [scale=0.2]{transformsTo};
	\pic [left=1.5cm of arrow1-leftEnd] {red-ii={e0}{e0}{e0}{e0}{$e_i$}};
	\pic [right=1.5cm of arrow1-rightEnd] {red-ii-single={e0}{$e_{i+1}$}};
	\pic (arrow2) [scale=0.2,right=3.5cm of arrow-rightEnd]{transformsTo};
	\pic [right=1.5cm of arrow2-rightEnd] {red-ii={e0}{e1}{e0}{e0}{$e_i$}};
\end{tikzpicture}
\caption{Reduction of type II (and its reverting). The matched edges are marked by thick lines.\label{fig:red-ii}}
\end{figure}

\subsection{$\Oh(n\log^2{n})$ time algorithm with fully dynamic connectivity}

Diks and Stańczyk presented a faster algorithm for the perfect matching problem in bridgeless cubic multigraphs by replacing a fully
dynamic $2$-edge-connectivity structure by a fully dynamic connectivity structure. Such a structure supports checking if two given vertices of the
multigraph are in the same connected component, and allows adding and removing edges. This turns out to be possible by maintaining 
a spanning tree of the current multigraph. When performing a reduction of type I on a single edge $e$, they remove all edges incident to any endpoint
of $e$ from both the fully dynamic connectivity structure and the spanning tree. Checking which pairs of vertices adjacent to endpoints of $e$
are in the same connected component and how they are connected allows to check if we have to perform straight or crossing reduction of type I to obtain
a bridgeless cubic multigraph as well. The spanning tree is maintained in the link-cut tree, so the total running time is dominated by the running time of
the fully dynamic connectivity structure. Originally, the algorithm used either the structure of Holm, de Lichtenberg and Thorup~\cite{HolmEtAl2001}, which
works in $\Oh(\log^2{n})$ amortized time per operation, or the randomized variant presented by Thorup~\cite{Thorup2000} which works in $\Oh(\log{n}(\log{\log{n}})^3)$
expected amortized time per operation. 
However, now the fastest known fully dynamic connectivity structures are by Wulff-Nilsen~\cite{WulffNilsen2012} with $\Oh(\log^2{n}/\log{\log{n}})$ amortized time per
operation, or the structure given by Huang, Huang, Kopelowitz, Pettie, and Thorup~\cite{HuangEtAl2023} with $\Oh(\log{n}(\log{\log{n}})^2)$ amortized
expected time per operation. Hence, the perfect matching problem in bridgeless cubic multigraphs can be solved in $\Oh(n\log^2{n}/\log{\log{n}})$ deterministic
time, or $\Oh(n\log{n}(\log{\log{n}})^2)$ expected time.

\section{Outline of our algorithm}

We give an outline of our algorithm below. It is based on the algorithm given by Biedl, Bose, Demaine, and Lubiw~\cite{BiedlEtAl2001}, but it does not need a
fully dynamic $2$-edge-connectivity structure. Let $G_0$ be the input multigraph, which is bridgeless and cubic. 
We proceed in iterations that construct the sequence $G_0$, $G_1$, \ldots, $G_k$ as follows.

\begin{algorithm}
\caption{}
\begin{algorithmic}
\State $e_{0} \gets$ any edge of $G_0$
\State $T_0\gets$ any spanning tree of $G_0$
\For {$e\in E(G_0)\setminus E(T_0)$}
	\State Set $\cover_0(f)\gets e$ for every $f\in E(T_0)$ on the path in $T_0$ connecting both endpoints of $e$
\EndFor
\For {$i= 0$ to $n/2-2$}
	\If {$e_i$ is incident to a single edge $e$ of $G_i$}
		\State Obtain $G_{i+1}$ by performing a reduction of type I on edge $e$ of $G_i$
	\Else
		\State Obtain $G_{i+1}$ by performing a reduction of type II on a double edge of $G_i$ incident to $e_i$
	\EndIf
	\State Obtain spanning tree $T_{i+1}$ of $G_{i+1}$ from $T_i$
\EndFor
\State $M_k\gets\{e\}$ for some $e\in E(G_k)\setminus\{e_k\}$
\For {$i= n/2-2$ downto $0$}
	\State Obtain a perfect matching $M_i$ of $G_i$ from $M_{i+1}$ by reverting the corresponding reduction
\EndFor
\end{algorithmic}
\end{algorithm}

 Similarly as in the algorithm given by Diks and Stańczyk, for every $G_i$, we construct a spanning tree $T_i$ of $G_i$ as well.
Additionally, for every $e\in E(T_i)$ we maintain any edge from $E(G_i)\setminus E(T_i)$ which covers $e$ in $T_i$, denoted $\cover_i(e)$.
Notice that such an edge always exists, since $G_i$ is bridgeless.
The spanning tree $T_i$ and $\cover_i(e)$, for every $e\in E(T_i)$, are maintained in a link-cut tree as described in Section~\ref{sec:link-cut}.
Moreover, we maintain an
edge $e_i$, which will not belong to the found perfect matching $M_i$ as in the algorithm presented in Subsection~\ref{subsec:biedl}.
We will construct the spanning trees $T_0$, $T_1$, \ldots, $T_k$ during the execution of the algorithm, making sure to maintain the following invariant.

\begin{invariant0}\label{inv:cover}
For every $T_i$ and edge $e\in E(T_i)$, $\cover_i(e)$ is an edge of $E(G_i)\setminus E(T_i)$ that covers $e$ in $T_i$.
\end{invariant0}

\section{Details}
In this section we explain how to implement the reductions and update the maintained information during the execution of the algorithm. Moreover, we prove that
Invariant~\ref{inv:cover} is maintained. In Subsection~\ref{subsec:complexity} we present the time and space complexity analysis.

\subsection{Swap}
We first define our atomic operation \emph{swap} on an edge $e\in E(T_i)$. It consists in removing $e$ from $T_i$, adding $e'=\cover_i(e)$ to $T_i$,
and setting $\cover_i(f)=e$ for every edge $f$ covered by $e$ in the new $T_i$. We will be using swap operation as a black box. The following lemma
proves that performing a swap does not spoil the maintained information.

\begin{lemma}\label{lem:swap}
The swap operation maintains Invariant~\ref{inv:cover} for $T_i$.
\end{lemma}
\begin{proof}
Assume that we perform a swap on an edge $e$ and let $e'=\cover_i(e)$. Let $T^0$ be the tree $T_i$ before the swap and $T^1$ after the swap. Let $P$ be the path in
$T^0$ which connects both endpoints of $e'$. Since $e'$ covers exactly the edges of $P$ in $T^0$, $e\in E(P)$. Notice that $e$ covers all edges of $P'=(P\setminus e)\cup e'$
in $T^1$, so every $\cover_i(f)$, for $f\in E(P')$, is updated correctly. Consider any $f\in E(T^1)\setminus E(P')$. By construction, $f$ belongs to $E(T^0)$, so $f'=\cover_i(f)$ covers 
$f$ in $T^0$. Therefore, there exists a path $R$ in $T^0$ which connects both endpoints of $f'$ such that $f\in E(R)$. We claim that $f'$ covers $f$ in $T^1$ as well.
If $e\notin E(R)$, then $R$ is a path in $T^1$ so we are done. Hence, assume that $e\in E(R)$. Notice that $f'\neq e'$ since $e'$ covers only the edges of $E(P)$ in $T^0$.
We construct a path $R'$ in $T^1$ from $R$ by replacing its fragment which is contained in $P$ by going through edge $e'$ instead of $e$ (see Figure~\ref{fig:swap}).
Notice that $f\in E(R')$ since $f\notin E(P')$, so we are done.
\end{proof}

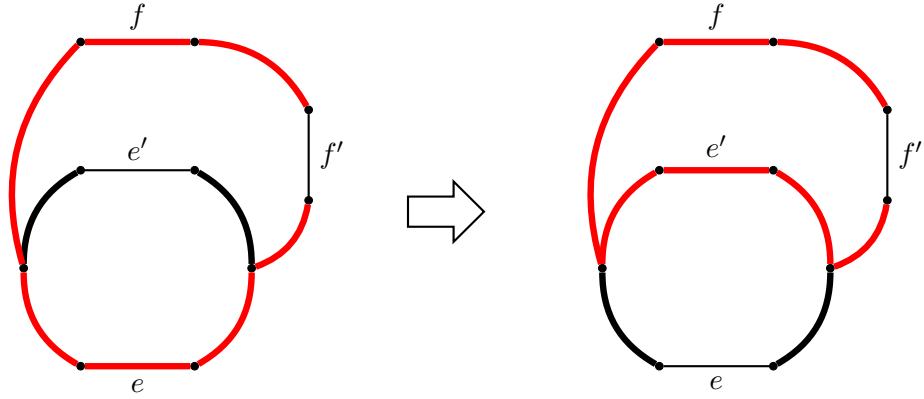
\begin{figure}[htpb]
\centering
\begin{tikzpicture}[transform shape]
	\pic (arrow) [scale=0.2]{transformsTo};
	\pic [left=3.5cm of arrow-leftEnd] {swap={e1}{e0}{red}{black}};
	\pic [right=3cm of arrow-rightEnd] {swap={e0}{e1}{black}{red}};
\end{tikzpicture}
\caption{Proof of Lemma~\ref{lem:swap}. The edges of $T_i$ are marked by thick lines and edges of, correspondingly, $R$ and $R'$ are marked by red lines. \label{fig:swap}}
\end{figure}

\subsection{Reductions of type I}
In this subsection we present how we implement a reduction of type I on a single edge $\{v,w\}$ of $G_i$. We use the notation from the statement of Theorem~\ref{thm:frink}. Let $A_i$ be the set of all edges of $G_i$ incident to edge $\{v,w\}$. Notice that $a$, $b$, $c$ and $d$ are not necessarily distinct, but they are different than $v$ and $w$ since $\{v,w\}$ is single. Moreover, $v\neq w$.

The main idea is the following. Before performing the reduction, we reduce the number of cases to consider by performing several swaps on some edges incident to $v$ or $w$. Recall that, by Lemma~\ref{lem:swap}, these swaps do not spoil the maintained information. Then, we perform either straight or crossing reduction of type I on $\{v,w\}$ depending on how the vertices $a$, $b$, $c$ and $d$ are connected in $(T_i\setminus v)\setminus w$. We obtain $T_{i+1}$ from $T_i$ by deleting the removed edges and adding some of the new edges.
This implicitly sets $\cover_{i+1}(e)=\cover_i(e)$ for every edge $e\in E(T_{i})\cap E(T_{i+1})$. 
Then, we update $\cover_{i+1}(e)$ for every edge $e$ covered with the edges of $E(G_{i+1})\setminus E(T_{i+1})$ in $T_{i+1}$ accordingly.
Finally, we set $e_{i+1}$ to be the new edge which is incident in $G_{i+1}$ to one of the endpoints of $e_i$ in $G_i$.

If $\{v,w\}\in E(T_i)$, we perform a swap on edge $\{v,w\}$. Hence, we can assume that $\{v,w\}\notin E(T_i)$. Moreover, if $A_i\subseteq E(T_i)$, we can perform a swap on at least one of the edges of $A_i$ without adding $\{v,w\}$ to $T_i$. Thus, we assume that either two or three edges of $A_i$ belong to $E(T_i)$. Furthermore, we can assume that, for every $e\in A_i\cap E(T_i)$, $\cover_i(e)\in A_i\cup\{\{v,w\}\}$, as otherwise we can perform a swap on such edge $e$.

If $|A_i\cap E(T_i)|=3$, we can assume that $\{d,w\}\notin E(T_i)$. We consider two subcases depending on how $a$, $b$, $c$ and $d$ are connected in $(T_i\setminus v)\setminus w$ (see Figure~\ref{fig:case3}):
\begin{itemize}
\item both $c$ and $d$ are connected to $a$ (or both to $b$), or 
\item $c$ and $d$ are connected to different vertices $a$ and $b$.
\end{itemize}
Notice that the first subcase cannot happen:
if both $c$ and $d$ are connected to $a$ then, since $\cover_i(\{b,v\})\notin A_i\cup\{\{v,w\}\}$, we could have performed a swap on $\{b,v\}$. Hence, we are left with the second subcase. Assume, without loss of generality, that $c$ is connected to $a$ and $d$ to $b$ in $(T_i\setminus v)\setminus w$. Then we perform a crossing reduction on $\{v,w\}$. Moreover, we add exactly one of the added edges to $T_{i+1}$.

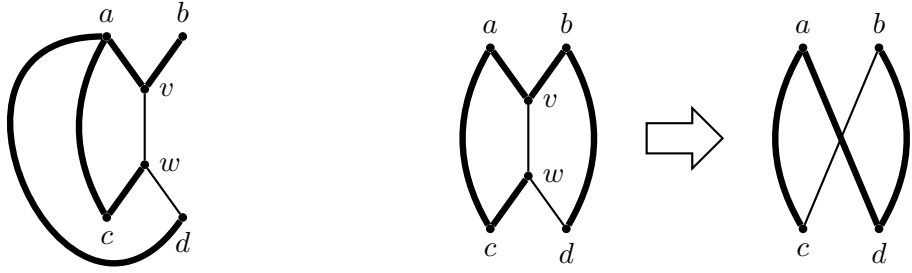
\begin{figure}[htpb]
\centering
\begin{minipage}{.49\textwidth}
	\centering
	\begin{tikzpicture}[transform shape]
		\pic {case3a};
	\end{tikzpicture}
\end{minipage}
\begin{minipage}{.49\textwidth}
	\begin{tikzpicture}[transform shape]
		\pic (arrow) [scale=0.2]{transformsTo};
		\pic [left=1.5cm of arrow-leftEnd] {case3b-1};
		\pic [right=1.5cm of arrow-rightEnd] {case3b-2};
	\end{tikzpicture}
\end{minipage}
\caption{The case when $|A_i\cap E(T_i)|=3$. The edges of $T_i$ and $T_{i+1}$ are marked by thick lines.\label{fig:case3}}
\end{figure}

If $|A_i\cap E(T_i)|=2$, assume that $\{a,v\}$ and $\{c,w\}$ belong to $E(T_i)$. Since $(T_i\setminus v)\setminus w$ is still connected in that case, we have the three following subcases depending on which pairs of the vertices $a$, $b$, $c$ and $d$ are connected in $(T_i\setminus v)\setminus w$ first (see Figure~\ref{fig:case2}).
Formally, we partition $\{a,b,c,d\}$ into two pairs $\{x,y\}$ and $\{x',y'\}$ such that the paths connecting $x$ with $y$
and $x'$ with $y'$ in $(T_i\setminus v)\setminus w$ are edge-disjoint (but not necessarily vertex-disjoint). Such edge-disjoint
paths always exists, for example it is straightforward to verify that paths with the smallest total length are edge-disjoint.
Then, we say that $x$ is connected to $y$ and $x'$ to $y'$.
\begin{itemize}
\item If $a$ is connected to $b$ and $c$ to $d$, we perform arbitrary reduction of type I on $\{v,w\}$.
\item If $a$ is connected to $c$ and $b$ to $d$, we perform crossing reduction on $\{v,w\}$.
\item If $a$ is connected to $d$ and $b$ to $c$, we perform straight reduction on $\{v,w\}$.
\end{itemize}
In all of these subcases we add no new edges to $T_{i+1}$. Notice that these subcases may overlap.

\begin{figure}[htpb]
\centering
\begin{minipage}{.49\textwidth}
	\begin{tikzpicture}[scale=1,transform shape]
		\pic (arrow) [scale=0.2]{transformsTo};
		\pic [left=2.3cm of arrow-leftEnd] {case2a-1};
		\pic [right=1cm of arrow-rightEnd] {case2a-2};
	\end{tikzpicture}
\end{minipage}
\begin{minipage}{.49\textwidth}
	\begin{tikzpicture}[scale=1,transform shape]
		\pic (arrow) [scale=0.2]{transformsTo};
		\pic [left=1.7cm of arrow-leftEnd] {case2b-1};
		\pic [right=1.7cm of arrow-rightEnd] {case2b-2};
	\end{tikzpicture}
	\centering
	\begin{tikzpicture}[scale=1,transform shape]
		\pic (arrow) [scale=0.2]{transformsTo};
		\pic [left=1.7cm of arrow-leftEnd] {case2c-1};
		\pic [right=1.7cm of arrow-rightEnd] {case2c-2};
	\end{tikzpicture}
\end{minipage}
\caption{The case when $|A_i\cap E(T_i)|=2$. The edges of $T_i$ and $T_{i+1}$ are marked by thick lines.\label{fig:case2}}
\end{figure}
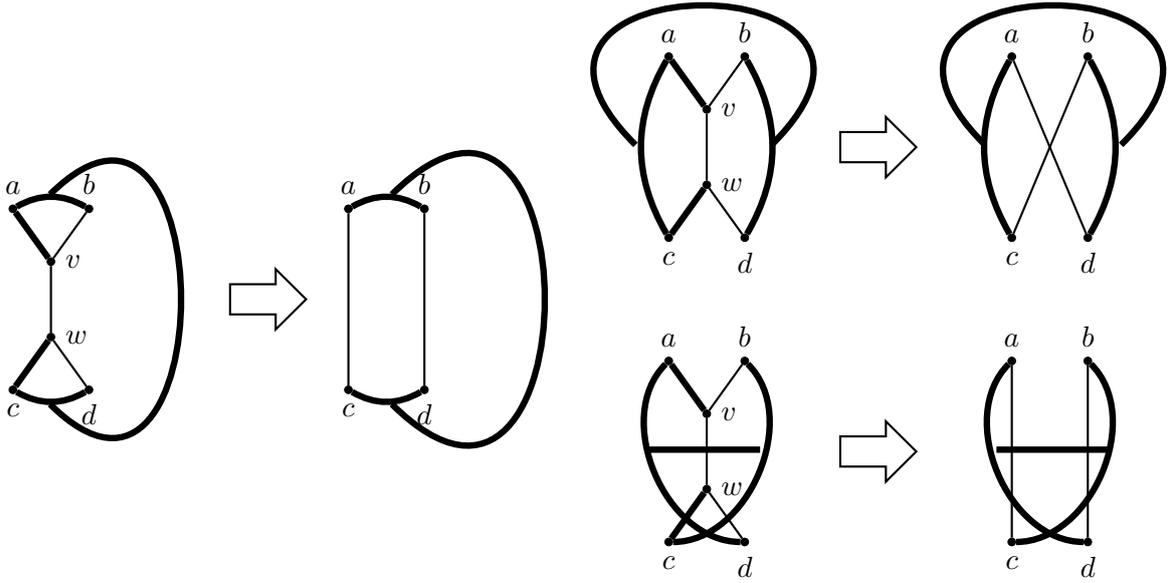

After performing every reduction of type I, we update the maintained information as follows.
For every edge $e\in E(G_{i+1})\setminus E(G_i)$ which does not belong to $T_{i+1}$, we set $\cover_{i+1}(f)=e$ for every edge $f$ on the path in $T_{i+1}$ connecting both endpoints of $e$. 

\begin{lemma}\label{lem:invariant-red-i}
Performing a reduction of type I maintains Invariant~\ref{inv:cover}.
\end{lemma}
\begin{proof}
Let $S_i$ and $S_{i+1}$ be the subgraphs of, respectively, $T_i$ and $T_{i+1}$ consisting of all edges which lie on some path in, respectively, $T_i$ and $T_{i+1}$ connecting some of the vertices $v$, $w$, $a$, $b$, $c$ and $d$. It is easy to check that every edge of $E(S_{i+1})$ is covered by some edge added to $G_{i+1}$ which does not belong to $T_{i+1}$.
Hence, $\cover_{i+1}(e)$ for every edge $e\in E(S_{i+1})$ is correct.

It is left to prove that $\cover_{i+1}(e)$ is correct for every edge $e\in E(T_{i+1})\setminus E(S_{i+1})$. By construction, $e\in E(T_i)\setminus E(S_i)$. Notice that $\cover_{i+1}(e)=\cover_i(e)$. We claim that $f=\cover_i(e)$ covers $e$ in $T_{i+1}$. First, we notice that $f$ belongs to $G_{i+1}$. This follows from an easy observation that every edge of $(A_i\cup\{\{v,w\}\})\setminus E(T_i)$ covers only some edges of $E(S_i)$ in $T_i$, so it cannot cover $e$ in $T_i$. Consider a path $P$ in $T_i$ which connects both endpoint of $f$. From definition of $f$, $e\in E(P)$. If $P$ does not contain any edges of $E(S_i)$, then $P$ is a path in $T_{i+1}$ as well, so $f$ covers $e$ in $T_{i+1}$. Otherwise, we construct a path $P'$ from $P$ by replacing its fragment which is contained in $S_i$ by a corresponding path in $S_{i+1}$ (see Figure~\ref{fig:invariant-red-i}). Since $e\notin E(S_i)$, $e\in E(P')$. Hence, $f$ covers $e$ in $T_{i+1}$.
\end{proof}

\begin{figure}[htpb]
\centering
\begin{tikzpicture}[scale=0.8,transform shape]
	\pic (arrow) [scale=0.2]{transformsTo};
	\pic [left=3cm of arrow-leftEnd,scale=1.25] {invariant={red={e0}{e1,red}{e1,red}{e1}{e0}}{red}{}};
	\pic [right=3cm of arrow-rightEnd,scale=1.25] {invariant={red-cross={e1,red}{e0}}{}{red}};
\end{tikzpicture}
\caption{The proof of Lemma~\ref{lem:invariant-red-i}. The edges of $T_i$ and $T_{i+1}$ are marked by thick lines. The edges of, correspondingly, $P$ and $P'$ are marked by red lines.\label{fig:invariant-red-i}}
\end{figure}
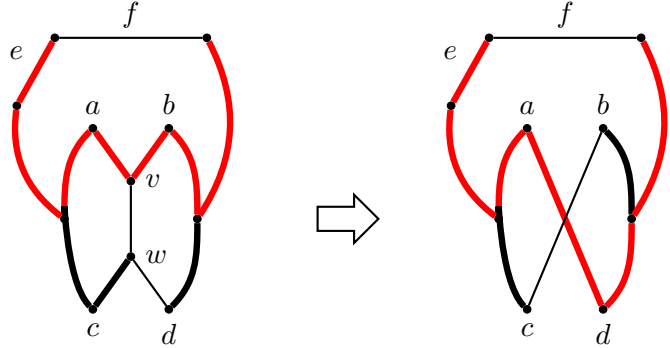

\subsection{Reductions of type II}
Now we consider the reduction of type II on a double edge incident to $e_i$ of $G_i$. Recall that $e_i$ is incident to two different double edges. Therefore, at least one of them, say $\{f_1,f_2\}$, contains some edge of $E(T_i)$ as otherwise $T_i$ would be disconnected. We perform a reduction of type II on double edge $\{f_1,f_2\}$ where $f_1=\{v,w\}=f_2$. Let $a\neq w$ be a neighbor of $v$, so $e_{i}=\{a,v\}$, and $b\neq v$ be a neighbor of $w$. Assume that $f_1\in E(T_i)$. Of course, $f_2\notin E(T_i)$ then.
First, we remove from $G_i$ edges  $e_i$, $f_{1}$, $f_{2}$ and $\{b,w\}$. If any of these edges is in $T_i$ it is not included in $T_{i+1}$.
Then, we create an edge $\{a,b\}$.
We consider the two following subcases (see Figure~\ref{fig:red-ii-analysis}).
\begin{itemize}
\item If both $\{a,v\}$ and $\{b,w\}$ belong to $E(T_i)$, we add $\{a,b\}$ to $T_{i+1}$
and set $\cover_{i+1}(\{a,b\})=\cover_i(\{a,v\})$.
\item If only one of $\{a,v\}$ or $\{b,w\}$ belongs to $E(T_i)$, say $\{a,v\}$, we identify $\{a,b\}$ with $\{b,w\}$. This guarantees that $\cover_{i+1}(e)$ is correct for every edge $e$ on the path in $T_{i+1}$ connecting $a$ and $b$.
\end{itemize}
It is easy to check that Invariant~\ref{inv:cover} is maintained after a reduction of type II.

\begin{figure}[htpb]
\centering
\begin{subfigure}{.49\textwidth}
	\centering
	\begin{tikzpicture}[transform shape]
		\pic (arrow) [scale=0.2]{transformsTo};
		\pic [left=1.5cm of arrow-leftEnd] {red-ii={e1}{e1}{e0}{e1}{$e_i$}};
		\pic [right=1.5cm of arrow-rightEnd] {red-ii-single={e1}{$e_{i+1}$}};
	\end{tikzpicture}
\end{subfigure}
\begin{subfigure}{.49\textwidth}
	\centering
	\begin{tikzpicture}[transform shape]
		\pic (arrow) [scale=0.2]{transformsTo};
		\pic [left=1.5cm of arrow-leftEnd] {red-ii={e1}{e1}{e0}{e0}{$e_i$}};
		\pic [right=1.5cm of arrow-rightEnd] {red-ii-single={e0}{$e_{i+1}$}};;
	\end{tikzpicture}
\end{subfigure}
\caption{Performing a reduction of type II. The edges of $T_i$ and $T_{i+1}$ are marked by thick lines.\label{fig:red-ii-analysis}}
\end{figure}
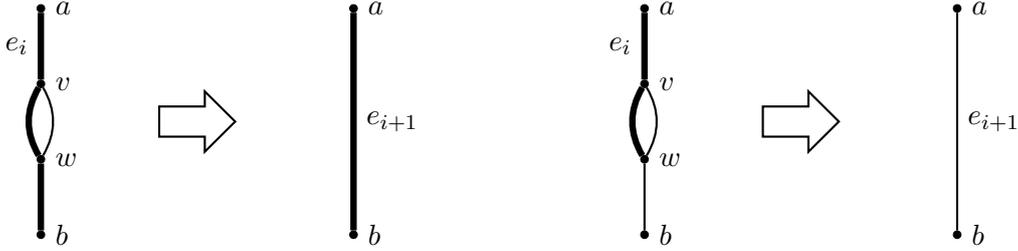

\subsection{Complexity analysis}\label{subsec:complexity}

The algorithm performs less than $n$ iterations of the for loop.
Throughout the execution of the algorithm, we maintain the spanning tree $T_i$ in a link-cut tree.
Additionally, we maintain the edges incident to any vertex of $G_i$ on a doubly-linked list.
Each edge maintains a single bit denoting whether it belongs to $T_i$.
Because the degree of every vertex of $G_i$ is constant, this allows us to find a single edge incident to $e_{i}$,
or choose a double edge $\{f_{1},f_{2}\}$ such that $f_{1}\in E(T_i)$, in constant time.

It is straightforward to verify that a swap operation can be implemented in constant number of operations on the link-cut tree.
To implement a reduction of type I, we first need a constant number of swap operations.
Then, we need to distinguish between $|A_i \cap E(T_i) | = 3$ and $|A_i \cap E(T_i)| = 2$, which is easy
by inspecting the bits maintained by the edges in $A_i$. In the latter case, we need to find
a partition of $\{a,b,c,d\}$ into two pairs $\{x,y\}$ and $\{x',y'\}$ such that the corresponding paths
in $(T_i\setminus v)\setminus w$ are edge-disjoint. To this end, we can check all 3 possibilities, and for each
of them test if the corresponding paths in $T_i$ are edge-disjoint in $\Oh(\log n)$ time.
Finally, after deciding whether we should apply a crossing or a straight reduction, we construct $G_{i+1}$ from $G_i$ by removing vertices $v$ and $w$
and their incident edges (and possibly from $T_i$), and adding the appropriate two edges,
and (in the case $|A_i \cap E(T_i) | = 3$) add one of them to $T_{i+1}$. Then, we update the cover values.
Overall, this takes $\Oh(\log n)$ time.

To implement a reduction of type II, we obtain $G_{i+1}$ from $G_i$ by removing vertices $v$ and $w$ and their incident edges
(and possibly from $T_i$), and adding edge $\{a,b\}$. In the first subcase, we update the cover value of the new edge.
In the second subcase, we need to implicitly update the cover value of every edge $e$ such that $\cover_i(e)=\{b,w\}$
to the new edge. To this end, we think that each edge $e=\{u,v\}$ is an object that stores the endpoints
$u$ and $v$. Then, $\cover_i(e)$ returns a pointer to the corresponding object.
When creating a new edge, we create a new object. However, in the second subcase we reuse the object corresponding
to the edge $\{b,w\}$, and modify its endpoints.

To reverse the reductions, we maintain the current matching $M_i$. Each edge stores a single bit denoting whether it belongs
to $M_i$. Then, reverting a reduction of type I takes only constant time by inspecting one of the new edges and checking if
it belongs to $M_i$. Depending on the case, we appropriately update $M_i$. 
Reverting a reduction of type II is even simpler, as we always add one copy of the double edge to $M_i$,
and possibly need to restore the object corresponding to the edge $\{b,w\}$.
For both types, we remove the new edges and add back the removed vertices and edges.

The overall time complexity is $\Oh(n\log n)$, and the algorithm uses $\mathcal{O}(n)$ space.

\section{Conclusions}
We presented a simple algorithm for the perfect matching problem in bridgeless cubic multigraphs, which works in $\Oh(n\log{n})$ deterministic time.
As opposed to the previous algorithms, it does not use any complex fully dynamic ($2$-edge-)connectivity structure. The natural open
question is to further improve the time complexity.

Another open problem is to apply a similar approach to the \emph{unique perfect matching problem} in sparse graphs.
It consists in checking if a given graph admits exactly one perfect matching, and finding it if so. The fastest known
deterministic algorithm for this problem was given by Gabow, Kaplan and Tarjan~\cite{GabowEtAl2001}, and takes $\Oh(n(\log{n})^2(\log{\log{n}})^2)$ time
when using the fastest fully dynamic $2$-edge-connectivity structure given by Holm, Rotenberg and Thorup~\cite{HolmEtAl2018}.
Note that the unique perfect matching problem can be solved in optimal linear time in dense graphs by using the decremental dynamic $2$-edge-connectivity structure given by Aamand
et al.~\cite{AamandEtAl2023}.

\bibliography{bib}

\end{document}